\documentclass[letterpaper, 10 pt, conference]{ieeeconf}  
\IEEEoverridecommandlockouts
\usepackage[normalem]{ulem}
\usepackage{amsmath}
\usepackage{amsfonts}
\usepackage{graphicx}
\usepackage{xcolor}
\usepackage{figsize}
\usepackage{upgreek}
\usepackage{amssymb}
\usepackage{algorithmic}
\usepackage[linesnumbered,ruled,vlined]{algorithm2e}
\usepackage{bm}
\usepackage{siunitx}
\usepackage{mathtools}
\usepackage[caption=false,font=footnotesize]{subfig}
\usepackage{booktabs}
\usepackage{adjustbox}  

\usepackage{amsthm}

\newtheorem{definition}{\bf Definition}
\newtheorem{theorem}{\bf Theorem}
\newtheorem{lemma}{\bf Lemma}

\definecolor{darkpink}{RGB}{220, 50, 120}

\graphicspath{{fig/}}

\title{\LARGE \bf
Interaction-Aware Predictive Environmental Control Barrier Function for Emergency Lane Change}

\author{Ying Shuai Quan$^{1}$, Paolo Falcone$^{1}$, Jonas Sj$\ddot{\text{o}}$berg$^{1}$
\thanks{This work is funded by the Fordonsstrategisk
Forskning och Innovation program of Vinnova under the grant 2018-05005.
		}
\thanks{$^{1}$ Ying Shuai Quan, Paolo Falcone and Jonas Sj$\ddot{\text{o}}$berg are with the Mechatronics group at the Department of Electrical Engineering, Chalmers University of Technology, Gothenburg, Sweden {\tt\footnotesize \{quany,paolo.falcone,jonas.sjoberg\}@chalmers.se }}
}

\begin{document}

\maketitle
\thispagestyle{empty}
\pagestyle{empty}

\begin{abstract}
Safety-critical motion planning in mixed traffic remains challenging for autonomous vehicles, especially when it involves interactions between the ego vehicle (EV) and surrounding vehicles (SVs). In dense traffic, the feasibility of a lane change depends strongly on how SVs respond to the EV motion. This paper presents an interaction-aware safety framework that incorporates such interactions into a control barrier function (CBF)-based safety assessment. The proposed method predicts near-future vehicle positions over a finite horizon, thereby capturing reactive SV behavior and embedding it into the CBF-based safety constraint. To address uncertainty in the SV response model, a robust extension is developed by treating the model mismatch as a bounded disturbance and incorporating an online uncertainty estimate into the barrier condition. Compared with classical environmental CBF methods that neglect SV reactions, the proposed approach provides a less conservative and more informative safety representation for interactive traffic scenarios, while improving robustness to uncertainty in the modeled SV behavior.
\end{abstract}

\section{Introduction}

Autonomous driving in mixed traffic requires safety-critical decision
making in environments where surrounding vehicles (SVs) react to the
ego vehicle (EV). This coupling is particularly important in emergency
lane changes: when a suddenly appearing road user (RU) blocks the
current lane, the feasibility of the maneuver depends not only on the
current gap in the target lane, but also on whether the SV yields to the
EV. Therefore, safety assessment based only on instantaneous geometry or
non-reactive environment models can be overly conservative or even
misleading. This is illustrated in Fig.~\ref{fig:scenario}, where accounting for the SV reaction opens a feasible gap for lane change, whereas a nominal non-reactive prediction deems the maneuver infeasible.

Control barrier functions (CBFs) have been widely used in autonomous
driving for safety-critical control, since they provide
forward-invariance guarantees for safe sets while remaining compatible
with online optimization-based control
synthesis~\cite{ames2016control,jankovic2018robust}. Environmental CBFs
(ECBFs) extend this framework to dynamic environments by reasoning over
joint EV–environment dynamics~\cite{molnar2022safety}. However, many
existing formulations, including ECBFs, still model surrounding traffic
as exogenous or non-interactive, which is insufficient for  interactive mixed-traffic scenarios.

To capture such interaction in mixed traffic, the SV response must be
modeled rather than treated as exogenous. Car-following models provide a
natural description of SV longitudinal behavior and can be extended to
account for anticipatory reactions to the EV~\cite{brito2022learning}. In practice, however, the
true SV behavior is uncertain and may deviate from the nominal model.
Such mismatch affects both the magnitude of the SV response and the
timing at which the interaction becomes active. A practical
interaction-aware safety framework should therefore incorporate both an
explicit SV response model and robustness to model uncertainty, while
remaining computationally tractable online.

\begin{figure}[t]
  \centering
  \includegraphics[width=\columnwidth]{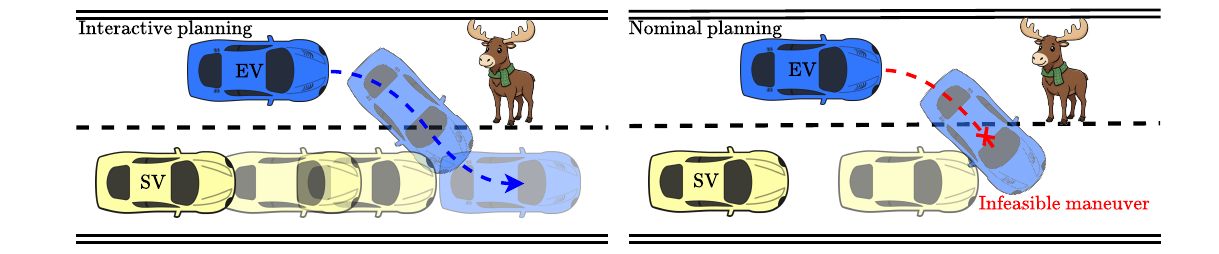}
  \caption{Interactive versus nominal planning in an obstacle-avoidance lane-change scenario. Faded yellow and blue vehicles indicate predicted future positions of the SV and EV. The blue dotted curve shows interactive planning, where the SV is expected to decelerate and open a feasible gap for lane change. The red dotted curve shows nominal planning, where the SV is assumed to maintain constant velocity, making the maneuver infeasible.}
  \label{fig:scenario}
\end{figure}

Motivated by this, we propose an
\emph{interaction-aware predictive environmental control barrier
function} (IECBF) for emergency lane change. The method incorporates a
joint EV--SV model into the ECBF framework, with P-IDM used to capture
anticipatory SV interaction, and constructs a predictive barrier
inspired by~\cite{breeden2022predictive} that
evaluates the most critical EV--SV configuration over a finite horizon.
To address model uncertainty, we further develop a robust extension that
combines observer-based correction for acceleration mismatch with a
multi-model rollout for gateway uncertainty.
The main contributions of this paper are: 1) We formulate an interaction-aware joint EV--SV model within an
ECBF framework and construct a predictive IECBF that evaluates the
most critical EV--SV configuration over a finite horizon.
2) We develop a robust extension under SV-model uncertainty by
combining observer-based correction for acceleration mismatch with a
multi-model rollout for gateway uncertainty.
3) We demonstrate in simulation that compared with a baseline
ECBF using a non-reactive SV model, the proposed method improves
feasibility and reduces conservatism in interactive emergency
lane-change scenarios.


\section{Preliminaries}
\subsection{Scenario Description}
\label{sec:sim_setup}
The road is modeled as a straight two-lane segment. We consider an
emergency lane-change scenario in which the EV must leave its initial
lane due to a suddenly appearing RU and merge into a target lane
occupied by a single SV. Throughout the maneuver, the SV is assumed to
remain in its lane and react only through longitudinal braking or
acceleration.

The objective is to design the EV control input to guarantee collision
avoidance while accounting for the SV’s reactive behavior through the
proposed IECBF framework. CBF constraints enforce safety with respect
to both the SV and the suddenly appearing RU, modeled as a temporary
static obstacle during lane change, while CLF-based constraints encode
the EV’s nominal tracking objective.

\subsection{EV Dynamics}
\label{subsec:ev_dynamics}
We consider continuous-time dynamics. Let $x(t)\in\mathbb{R}^{n_x}$ and $u(t)\in\mathbb{R}^{n_u}$ denote the EV state and control input as: 
\begin{equation}
    x :=
    \begin{bmatrix}
        X_e & Y_e & \psi_e & v_e
    \end{bmatrix}^\top,
    u:=\begin{bmatrix}a_e& \delta_e\end{bmatrix}^\top, 
\end{equation}
where $(X_e,Y_e)$ denotes the position of the EV center of gravity (CG) in the inertial frame,
$\psi_e$ is the heading angle, $v_e$ is the longitudinal speed, $a_e$ is the EV longitudinal acceleration and $\delta_e$ is the steering angle.
The admissible input set is
$\mathcal U \coloneqq \{u=[a_e,\delta_e]^\top\in\mathbb R^2:
\delta_e\in[\delta_{e,\min},\delta_{e,\max}]\}$.
Under a small-angle assumption,
the EV dynamics is written as
\begin{equation}
    \dot x = f(x) + g(x)u,
    \label{eq:ev_dyn}
\end{equation}
with 
\begin{equation*}
    f(x) =
    \begin{bmatrix}
        v_e\cos\psi_e \\
        v_e\sin\psi_e \\
        0 \\
        0
    \end{bmatrix},
    g(x) =
    \begin{bmatrix}
        0 & -v_e\sin\psi_e \\
        0 & \;\;v_e\cos\psi_e \\
        0 & \dfrac{v_e}{\ell_r} \\
        1 & 0
    \end{bmatrix},
\end{equation*}
where $\ell_r$ is the distance from the vehicle CG to the rear axle.

\subsection{CLF Design for Lane Change}
\label{sec:clf_lane_change}

We employ standard control Lyapunov function (CLF) constraints to encode
the nominal lane-change objective~\cite{ames2016control}.
Let $Y_{\rm target}$ denote the center line of the target lane. During the
lane-change maneuver, it serves as the reference for the EV lateral position
and heading. We define
\begin{align}\label{eq:CLFs}
    V_y(x) &:= \bigl(Y_e - Y_{\rm target}\bigr)^2, \quad
    V_\psi(x) := \psi_e^2,
\end{align}
where $V_y$ penalizes lateral deviation from the target-lane center, while
$V_\psi$ penalizes heading deviation from the road direction. The corresponding
CLF constraints are imposed as soft constraints in a quadratic programming (QP) problem:
\begin{align}
    L_f V_y(x) + L_g V_y(x)\,u
    &\le
    -\alpha_y\bigl(V_y(x)\bigr) + \delta_y,
    \label{eq:clf_y_constraint}
    \\
    L_f V_\psi(x) + L_g V_\psi(x)\,u
    &\le
    -\alpha_\psi\bigl(V_\psi(x)\bigr) + \delta_\psi,
    \label{eq:clf_psi_constraint}
\end{align}
where $\alpha_{y,\psi}$ are class-$\mathcal{K}$ functions and
$\delta_{y,\psi}\geq 0$ are slack variables. These soft
constraints promote lateral tracking and heading regulation whenever
compatible with safety.

\subsection{Baseline ECBF Formulation}
\label{sec:baseline_ecbf}

This subsection presents a baseline ECBF design for the lane-change scenario using an EV–SV joint-state representation equivalent to the augmented-state form in \cite{molnar2022safety}, with the SV modeled as non-reactive at constant velocity.

\subsubsection{Baseline joint EV--SV model}
Let the SV state be
$x_s \coloneqq
    \begin{bmatrix}
        X_s & Y_s & v_s
    \end{bmatrix}^\top$,
where $(X_s,Y_s)$ denotes the position of the SV CG in the inertial frame, and $v_s$ denotes its longitudinal speed. The joint EV--SV
state is then defined as
\begin{equation}\label{eq:zstate}
    z \coloneqq
    \begin{bmatrix}
        x^\top & x_s^\top
    \end{bmatrix}^\top.
\end{equation}
With SV's nominal dynamics
given by
$\dot x_s = \bar f_s(x_s)$,  
    $\bar f_s(x_s) \coloneqq
    \begin{bmatrix}
        v_s \\
        \mathbf{0}_{2\times 1}
    \end{bmatrix}$
and combining with \eqref{eq:ev_dyn}, the
nominal joint dynamics are
\begin{equation}\label{eq:joint_sys_baseline}
    \dot z = \bar F(z) + \bar G(z)u, \bar F(z) \coloneqq
    \begin{bmatrix}
        f(x) \\
        \bar f_s(x_s)
    \end{bmatrix},
    \bar G(z) \coloneqq
    \begin{bmatrix}
        g(x) \\
        \mathbf{0}_{3\times 2}
    \end{bmatrix}.
\end{equation}

\subsubsection{ECBF Design}
\label{subsec:ecbf_baseline}

For the baseline joint system~\eqref{eq:joint_sys_baseline}, define the
safe set
    $\mathcal{C}_z
    :=
    \left\{
        z \in \mathcal{Z} : H(z) \ge 0
    \right\},$
where $H:\mathcal{Z}\to\mathbb{R}$ is continuously differentiable.

\begin{definition}[ECBF]
\label{def:ecbf_baseline}
A continuously differentiable function $H:\mathcal Z\to\mathbb R$ is
called an ECBF for \eqref{eq:joint_sys_baseline} if there exists a
class-$\mathcal{K}$ function $\alpha_H$ such that, for all
$z \in \mathcal{C}_z$,
\begin{equation}
    \sup_{u\in\mathcal U}
    \bigl(
        L_{F} H(z) + L_G H(z)\,u
    \bigr)
    \ge
    -\alpha_H\bigl(H(z)\bigr).
    \label{eq:ecbf_def_cond_baseline}
\end{equation}
\end{definition}

If a locally Lipschitz controller $u:\mathcal C_z\to\mathcal U$
satisfies
\begin{equation}
    L_{F} H(z) + L_G H(z)\,u(z)
    \ge
    -\alpha_H\bigl(H(z)\bigr),
    \quad
    \forall z \in \mathcal{C}_z,
    \label{eq:ecbf_forward_invariance_baseline}
\end{equation}
then $\mathcal C_z$ is forward invariant~\cite{molnar2022safety}.

To enforce collision avoidance between the EV and the SV, we choose
$H_{\mathrm{SV}}(z)$ based on the relative EV--SV position expressed in
the EV body frame, let
$\begin{bmatrix}
        \Delta X \\ \Delta Y
    \end{bmatrix}
    =
    R(\psi_e)^\top
    \begin{bmatrix}
        X_s - X_e \\ Y_s - Y_e
    \end{bmatrix}$
with $R(\psi_e)
    =
    \begin{bmatrix}
        \cos\psi_e & -\sin\psi_e \\
        \sin\psi_e & \cos\psi_e
    \end{bmatrix}$.
Here, $\Delta X$ and $\Delta Y$ denote the longitudinal and lateral
gaps in the EV body frame, respectively.
We then define the ECBF candidate as
\begin{equation}
    H_{\mathrm{SV}}(z)
    =
    a_\mathrm{SV} \Delta X^2
    +
    b_s \Delta Y^2
    - 1,
    \label{eq:H_baseline}
\end{equation}
where $a_\mathrm{SV} = 1/(r^\mathrm{SV}_a)^2$ and $b_\mathrm{SV} = 1/(r^\mathrm{SV}_b)^2$, with $r^\mathrm{SV}_a$ and $r^\mathrm{SV}_b$
denoting the longitudinal and lateral semi-axes of the ellipsoidal
safety envelope. 
Then with $\alpha_{\mathrm{SV}}$ as a class-$\mathcal{K}$ function,
the baseline ECBF constraint is
\begin{equation}
    L_{\bar F} H_{\mathrm{SV}}(z) + L_{\bar G} H_{\mathrm{SV}}(z)\,u
    \ge
    -\alpha_{\mathrm{SV}}\!\bigl(H_{\mathrm{SV}}(z)\bigr).
    \label{eq:ecbf_cond_baseline}
\end{equation}
\vspace{-0.5cm}
\subsection{RU ECBF}
\label{subsec:static_obstacle_cbf}

The suddenly appearing RU is modeled as a static obstacle at fixed
position $(X_{\mathrm{RU}},Y_{\mathrm{RU}})$ over the short emergency
lane change, which is a special case of the ECBF framework with constant
environment state.
We define the RU barrier as
\begin{equation}
    H_{\mathrm{RU}}(x)
    \coloneqq
    a_{\mathrm{RU}}(X_e - X_{\mathrm{RU}})^2
    + b_{\mathrm{RU}}(Y_e - Y_{\mathrm{RU}})^2 - 1,
    \label{eq:H_RU}
\end{equation}
where $a_{\mathrm{RU}} = 1/(r_a^{\mathrm{RU}})^2$ and
$b_{\mathrm{RU}} = 1/(r_b^{\mathrm{RU}})^2$.
Since the RU position is constant, $H_{\mathrm{RU}}$ depends
only on the EV state, and its Lie derivatives are taken with respect to
the EV dynamics \eqref{eq:ev_dyn}. The corresponding ECBF constraint is
\begin{equation}
    L_f H_{\mathrm{RU}}(x) + L_g H_{\mathrm{RU}}(x)\,u
    \ge -\alpha_{\mathrm{RU}}\!\bigl(H_{\mathrm{RU}}(x)\bigr).
    \label{eq:cbf_RU_cond}
\end{equation}
where $\alpha_{\mathrm{RU}}$ is a class-$\mathcal{K}$ function.
Since $a_e$ does not appear in \eqref{eq:cbf_RU_cond}, the constraint is
affine only in $\delta_e$ and reduces to
\begin{equation}
    [L_g H_{\mathrm{RU}}(x)]_2\,\delta_e
    \ge
    -\alpha_{\mathrm{RU}}\!\bigl(H_{\mathrm{RU}}(x)\bigr)
    - L_f H_{\mathrm{RU}}(x).
    \label{eq:delta_cbf_halfspace}
\end{equation}

\section{Interaction-aware ECBF Design}
\label{sec:iecbf_design}

\subsection{SV Dynamics via P-IDM}
\label{subsec:sv_idm}

The proposed method models SV dynamics based on the P-IDM car-following model in \cite{brito2022learning}. Unlike standard IDM, which assigns the leader based on the candidate vehicle’s current lateral position and thus reacts only after lane intrusion, P-IDM uses the predicted lateral position at a short horizon $N_{\mathrm{p}}$ ahead. This better captures anticipatory SV behavior in lane-change scenarios, where the candidate vehicle is the EV.

To encode this leader-selection logic compactly, we introduce the gate
variable
\begin{equation}
    \omega \coloneqq
    \begin{cases}
        1, & \text{if } X_e > X_s
        \text{ and }
        \left| \tilde{Y}_{e,N_{\mathrm{p}}} - Y_s \right| < c,\\
        0, & \text{otherwise,}
    \end{cases}
    \label{eq:pidm_gate}
\end{equation}
where $c$ is a cooperation threshold that determines the SV's
willingness to yield. Thus, $\omega=1$ means that the EV is treated as
the SV's leader based on its predicted lane intrusion.
When the EV is selected as leader, the SV longitudinal acceleration is
computed using the standard IDM~\cite{treiber2000congested},
\begin{equation}
    \begin{split}
        a_{\mathrm{idm}}\coloneqq
        a_{\max}\!\left[
        1 - \left(\frac{v_s}{v^*}\right)^{\!4}
          - \left(\frac{s^*\!\left(v_s,\,\Delta v_s\right)}
                      {\Delta x_s}\right)^{\!2}
    \right],
    \end{split}
    \label{eq:idm_accel}
\end{equation}
where $v^*$ is the desired speed, $a_{\max}$ is the maximum
acceleration, $\Delta x_s = X_e - X_s$ is the bumper-to-bumper gap
between the EV and the SV, and $\Delta v_s = v_s - v_e$ is the
relative approach speed. The desired minimum gap is given by
$s^*(v_s,\Delta v_s)
    = s_0 + v_s\,T_{\rm idm}
      + \frac{v_s\,\Delta v_s}{2\sqrt{a_{\max} b}}$,
where $s_0$ is the standstill distance, $T_{\rm idm}$ is the desired
time headway, and $b$ is the comfortable deceleration.
%
$a_{\mathrm{free}}
    \coloneqq
    a_{\max}\!\left[1-\left(\frac{v_s}{v^*}\right)^4\right]$.
Then the resulting P-IDM acceleration can be written as
\begin{equation}
    a_{\mathrm{pidm}}
    \coloneqq
    \omega\, a_{\mathrm{idm}}(\Delta x_s,\, v_s,\, \Delta v_s)
    + (1-\omega)\, a_{\mathrm{free}}(v_s).
    \label{eq:a_pidm_joint}
\end{equation}
Accordingly, the SV dynamics in the proposed model are
\begin{equation}
    \dot x_s = f_s(x,x_s),
    \qquad
    f_s(x,x_s) \coloneqq
    \begin{bmatrix}
        v_s \\
        0 \\
        a_{\mathrm{pidm}}
    \end{bmatrix},
    \label{eq:sv_dyn}
\end{equation}
where the dependence on the EV state enters through
$a_{\mathrm{pidm}}$.

\subsection{Interaction-Aware Joint EV--SV Model}
\label{subsec:physics_joint_world_model}

To explicitly capture the interaction between the EV and SV, we
combine \eqref{eq:ev_dyn} with  \eqref{eq:sv_dyn}. With the joint state \eqref{eq:zstate}, the interaction-aware joint dynamics take the control-affine form
\begin{equation}
    \dot z = F(z) + G(z)u,
    \label{eq:physics_world_model_ct}
\end{equation}
where
$F(z) \coloneqq
    \begin{bmatrix}
        f(x) \\
        f_s(x,x_s)
    \end{bmatrix},
    G(z) \coloneqq
    \begin{bmatrix}
        g(x) \\
        \mathbf{0}_{3\times 2}
    \end{bmatrix}$.
Here, $F(z)$ represents the
natural evolution of the combined EV–SV state, including the SV
response, and $G(z)$ specifies how the EV input
$u$ enters the joint dynamics.


\subsection{IECBF Constraints for SV}
\label{subsec:cbf_constraints}
In the proposed method, safety assessment is extended over a short prediction horizon under the interaction-aware SV response, since interaction effects may not be immediately visible but emerge as the SV reacts to the EV. Accordingly, we use two ECBF constraints: a current-position ECBF for present-state safety and a predictive ECBF over the finite horizon under the interaction-aware joint model.
\subsubsection{One-Step IECBF}
\label{subsubsec:cbf_current}
We use the same SV barrier function \(H^{\mathrm{curr}}_{\mathrm{SV}}(z) = H_{\mathrm{SV}}(z)\) as in the
baseline ECBF \eqref{eq:H_baseline}. From here on, we write
\(H^{\mathrm{curr}}_{\mathrm{SV}} \coloneqq H^{\mathrm{curr}}_{\mathrm{SV}}(z)\)
for brevity. The one-step ECBF constraint under
the joint dynamics \eqref{eq:physics_world_model_ct}
is then
\begin{equation}
    L_F H^{\mathrm{curr}}_{\mathrm{SV}} + L_G H^{\mathrm{curr}}_{\mathrm{SV}}\,u
    \ge
    -\alpha_{\mathrm{SV}}\!\bigl(H^{\mathrm{curr}}_{\mathrm{SV}}\bigr).
    \label{eq:iecbf_curr_cond}
\end{equation}
Compared with the baseline  \eqref{eq:ecbf_cond_baseline}, the term \(L_F H^{\mathrm{curr}}_{\mathrm{SV}}\)
now reflects the interaction-aware SV dynamics through
\(a_{\mathrm{pidm}}\). 

\subsubsection{Predictive ECBF}
\label{subsubsec:cbf_predictive}
To explicitly
account for near-future interaction, we construct a predictive barrier
by rolling out the joint model \eqref{eq:physics_world_model_ct} over a
finite horizon $N$.
\paragraph{Nominal rollout controller}
To roll out \eqref{eq:physics_world_model_ct}, we use a nominal steering
controller that drives the EV toward the target lane while satisfying
the RU ECBF constraint \eqref{eq:cbf_RU_cond}. The SV barrier is not
enforced in the rollout, since SV safety is handled by the outer
predictive ECBF and enforcing it in the inner rollout would bias the
predicted trajectory.

By \eqref{eq:delta_cbf_halfspace}, the RU ECBF together with the
steering bounds yields the feasible interval
$\delta_e \in [\delta_{\min}^{\mathrm{RU}},\,\delta_{\max}^{\mathrm{RU}}]$.
Within this interval, the nominal steering is chosen from the 1-D soft
CLF--CBF problem
\begin{equation*}
\label{eq:rollout_qp}
\begin{aligned}
\min_{\delta_e \in [\delta_{\min}^{\mathrm{RU}},\,\delta_{\max}^{\mathrm{RU}}]}
\quad
& H_\delta\,\delta_e^2
+ p_y \max(0,\phi_y)^2
+ p_\psi \max(0,\phi_\psi)^2,
\end{aligned}
\end{equation*}
where $\phi_y \coloneqq L_f V_y + [L_g V_y]_2\,\delta_e + \alpha_y(V_y),
\phi_\psi \coloneqq L_f V_\psi + [L_g V_\psi]_2\,\delta_e + \alpha_\psi(V_\psi)$, 
following the standard soft-constraint CLF--CBF construction
of~\cite{ames2016control}. Since this is a scalar soft CLF--CBF problem,
its minimizer admits a closed-form solution~\cite{jankovic2018robust},
so no online optimization is required in the rollout. Denoting the
resulting steering command at prediction step \(k\) by
\(\delta_{e,k}^{\mathrm{nom}}\), the nominal rollout input is
\begin{equation}
    u_k^{\mathrm{nom}}
    \coloneqq
    \begin{bmatrix}
        0 \\
        \delta_{e,k}^{\mathrm{nom}}
    \end{bmatrix}.
    \label{eq:u_nom}
\end{equation}
The nominal rollout sets \(a_e=0\) because the CLF and RU ECBF terms
used here determine only the steering input \(\delta_e\); \(a_e\) is
neither constrained nor optimized in the inner rollout.

\paragraph{Predictive barrier and constraint}
Let \(\{z_k\}_{k=0}^{N}\) denote the rollout from the current state
\(z_0\), where the EV is propagated using \eqref{eq:ev_dyn} under the
nominal control \eqref{eq:u_nom} over a horizon $N$.
At each rollout step, we evaluate \(H_{\mathrm{SV}}(z_k)\). The most
critical predicted step is
\begin{equation}
    k^* = \arg\min_{k \in \{0,\dots,N\}} H_{\mathrm{SV}}(z_k).
    \label{eq:kstar}
\end{equation}
We then define the predictive barrier as
\begin{equation}
    H^{\mathrm{pred}}_{\mathrm{SV}}(z_0)
    \coloneqq
    H_{\mathrm{SV}}(z_{k^*}),
    \label{eq:H_pred}
\end{equation}
i.e., the minimum predicted barrier value along the rollout.
Let \(J_{k^*} \coloneqq \partial z_{k^*}/\partial z_0\) denote the
rollout Jacobian from the current state to the critical predicted state.
Then
\begin{equation}
    \begin{aligned}
        L_F H^{\mathrm{pred}}_{\mathrm{SV}}(z_0)
        &=
        \nabla_z H_{\mathrm{SV}}(z_{k^*})\,J_{k^*}\,F(z_0), \\
        L_G H^{\mathrm{pred}}_{\mathrm{SV}}(z_0)
        &=
        \nabla_z H_{\mathrm{SV}}(z_{k^*})\,J_{k^*}\,G(z_0).
    \end{aligned}
    \label{eq:LFG_Hpred}
\end{equation}
Using the same class-\(\mathcal{K}\) function \(\alpha_{\mathrm{SV}}\)
as in \eqref{eq:iecbf_curr_cond}, and writing
\(H^{\mathrm{pred}}_{\mathrm{SV}} \coloneqq H^{\mathrm{pred}}_{\mathrm{SV}}(z_0)\)
for brevity, the predictive ECBF constraint enforced in
the outer QP is
\begin{equation}
    L_F H^{\mathrm{pred}}_{\mathrm{SV}} + L_G H^{\mathrm{pred}}_{\mathrm{SV}}\,u
    \ge
    -\alpha_{\mathrm{SV}}\!\bigl(H^{\mathrm{pred}}_{\mathrm{SV}}\bigr).
    \label{eq:iecbf_pred_cond}
\end{equation}
Constraint \eqref{eq:iecbf_pred_cond} complements the one-step
constraint \eqref{eq:iecbf_curr_cond}: the latter enforces safety at the
current state, while the former accounts for the most critical EV--SV
configuration over the prediction horizon under the interaction-aware
joint model.

\subsection{Nominal CLF-IECBF QP}
The nominal
CLF--IECBF QP is finally given as
\begin{equation}
\label{eq:robust_clf_cbf_qp}
\begin{aligned}
u^\star
&=
\arg\min_{\substack{u\in\mathcal U,\\ \delta_y,\,\delta_\psi}}
\;
\frac{1}{2}u^\top Q u
+\frac{1}{2}p_y\delta_y^2
+\frac{1}{2}p_\psi\delta_\psi^2 \\
\text{s.t.}\quad
&\eqref{eq:clf_y_constraint},\ \eqref{eq:clf_psi_constraint},\ \eqref{eq:cbf_RU_cond},\
\eqref{eq:iecbf_curr_cond},\
\eqref{eq:iecbf_pred_cond}.
\end{aligned}
\end{equation}

\section{Robust IECBF Under IDM Model Uncertainty}
\label{sec:robust_iecbf}

The IECBF in Section~\ref{sec:iecbf_design} uses a nominal P-IDM model for SV prediction. In practice, the true SV behavior may deviate due to uncertain driver characteristics. This mismatch appears in both the SV acceleration magnitude after interaction becomes active and the timing of interaction gate activation. To address both, this section extends the IECBF with observer-based correction for acceleration mismatch and a multi-style rollout for gateway uncertainty.

\subsection{Uncertainty Modeling}

\subsubsection{IDM Parameter Uncertainty}
\label{subsubsec:idm_param_uncertainty}

Let
\(\theta \coloneqq (a_{\max},\,v^*,\,s_0,\,T_{\mathrm{idm}},\,b)\)
denote the IDM parameter vector.
For later use in the multi-model rollout, let \(j\in\mathcal{J}\) index
a finite set of representative IDM parameter presets, with each
\(j\) associated with an IDM parameter vector \(\theta^{(j)}\).
Let \(\theta_{\mathrm{nom}}\) denotes the preset nominal parameter set.
Let \(\theta^*\) denote the unknown true IDM parameter vector.
Then the true SV longitudinal acceleration is modeled as
\begin{equation}
    \dot v_s
    =
    a_{\mathrm{idm}}(\theta_{\mathrm{nom}})
    + \Delta_\theta,
    \label{eq:sv_accel_uncertain}
\end{equation}
where \(\Delta_\theta
    \coloneqq
    a_{\mathrm{idm}}(\theta^*)
    -
    a_{\mathrm{idm}}(\theta_{\mathrm{nom}})\).
Accordingly, the interaction-aware joint dynamics
\eqref{eq:physics_world_model_ct} are rewritten as
\begin{align}
    \dot z = F(z) + G(z)u + B_s\,\Delta_\theta,
    \label{eq:perturbed_joint_dyn}
\end{align}
where
$B_s \coloneqq
\begin{bmatrix}
0 & 0 & 0 & 0 & 0 & 0 & 1
\end{bmatrix}^{\!\top}$.

We assume that the uncertainty and its rate of change are bounded:
\begin{equation}
    |\Delta_\theta| \le \delta_b,
    \qquad
    |\dot{\Delta}_\theta| \le \delta_L,
    \label{eq:uncertainty_bounds}
\end{equation}
for known constants \(\delta_b,\delta_L>0\), determined offline from
the admissible range of SV driver behavior.

The sensitivity of the predictive barrier to the acceleration
perturbation \(B_s\Delta_\theta\) at the current state is defined as
\(
\sigma
\coloneqq
\nabla_z H_{\mathrm{SV}}(z_{k^*})\,J_{k^*}\,B_s
=
[\nabla_z H_{\mathrm{SV}}(z_{k^*})\,J_{k^*}]_{v_s},
\)
i.e., the \(v_s\)-component of the propagated gradient. $\sigma$
quantifies how the current SV acceleration perturbation affects the
barrier at the critical step \(k^*\), and is used below in the robust
predictive IECBF constraint.


\subsubsection{Gateway activation uncertainty}
\label{subsubsec:gateway_uncertainty}


To capture gateway activation uncertainty, we consider a finite set of
representative P-IDM gateway parameter pairs
\(\varphi^{(i)}=(N^{(i)},c_{\mathrm{th}}^{(i)})\),
\(i\in\mathcal{I}\). Since gate activation depends on both the gateway
preset and the IDM model used in the rollout, we denote the
corresponding interaction gate for model pair \((i,j)\) by
\(\omega^{(i,j)}\), where \(i\in\mathcal{I}\) and \(j\in\mathcal{J}\).
Hence, even for the same \(\varphi^{(i)}\), different IDM presets
\(\theta^{(j)}\) can yield different gate activations and switching
times between free-road and car-following behavior along the predicted
trajectory.
This uncertainty is relevant only when \(\omega^{(i,j)}=1\), i.e.,
when the SV is predicted to react to the EV through the IDM
car-following model under \((\varphi^{(i)},\theta^{(j)})\). When
\(\omega^{(i,j)}=0\), the SV follows the free-road model
\(a_{\mathrm{free}}(\theta^{(j)})\), so the car-following mismatch
\(\Delta_\theta\) is inactive. Accordingly, the robust correction for
IDM parameter uncertainty is applied only while the corresponding gate
\(\omega^{(i,j)}\) is active.

\subsection{Gate-Activated Car-Following Observer}
\label{subsec:uncertainty_estimator}

Following~\cite{dacs2025robust}, we estimate the SV car-following
acceleration mismatch online from the measured SV speed \(v_s\), only
when the corresponding interaction gate is active, i.e.,
\(\omega^{(i,j)}=1\). 
%
Define the uncertainty estimate as
$\hat\Delta_\theta \coloneqq \lambda_s v_s - \xi$,
where \(\xi \in \mathbb{R}\) is the observer state and
\(\lambda_s > 0\) is the observer gain. During intervals with
\(\omega^{(i,j)}=1\), the observer dynamics are
\begin{equation}
    \dot\xi
    =
    \lambda_s\bigl(
        a_{\mathrm{idm}}(\theta^{(j)})
        + \hat\Delta_\theta
    \bigr).
    \label{eq:cf_estimator_dynamics}
\end{equation}
At each \(0\to1\) transition of \(\omega^{(i,j)}\), the observer state
is reset according to
$\xi \leftarrow \lambda_s v_s$
so that \(\hat\Delta_\theta = 0\) at the reset instant.

Let \(t_{\mathrm{cf}} \ge 0\) denote the elapsed time since the most
recent observer-gate activation, and define the estimation error as
\begin{equation}
    e_\theta \coloneqq \Delta_\theta - \hat\Delta_\theta.
    \label{eq:estimation_error}
\end{equation}
\begin{lemma}
\label{lem:error_bound}
Suppose the uncertainty \(\Delta_\theta\) satisfies
\eqref{eq:uncertainty_bounds}. Then, along each interval with
\(\omega^{(i,j)}=1\), the estimation error \(e_\theta\)
satisfies
$\dot e_\theta = \dot\Delta_\theta - \lambda_s e_\theta$,
and is bounded by
\begin{equation}
    |e_\theta(t)| \le
    \bar e(t_{\mathrm{cf}})
    \coloneqq
    \left(\delta_b - \frac{\delta_L}{\lambda_s}\right)
    e^{-\lambda_s t_{\mathrm{cf}}}
    + \frac{\delta_L}{\lambda_s},
    \label{eq:error_bound}
\end{equation}
for all \(t_{\mathrm{cf}} \ge 0\).
\end{lemma}

\begin{proof}
The result follows directly from Lemma~1 in~\cite{dacs2025robust}.
\end{proof}
Thus, the bound resets to \(\delta_b\) at each \(0\to1\) transition of
\(\omega^{(i,j)}\) and then decays monotonically to
\(\delta_L/\lambda_s\) during the corresponding interaction-active
interval. 
Define \(t_{\mathrm{conv}} \coloneqq 2/\lambda_s\), i.e., two time
constants after gate activation. Since the exponential term
in~\eqref{eq:error_bound} is then reduced to \(e^{-2}\approx 13.5\%\)
of its initial value, \(t_{\mathrm{conv}}\) is used as a conservative
proxy for observer convergence, and \(\hat\Delta_\theta\) is treated as
reliable only when \(t_{\mathrm{cf}} > t_{\mathrm{conv}}\).

\subsection{Robust Predictive ECBF Constraint}

\subsubsection{Case-Dependent Predictive Rollout}
\label{subsubsec:case_dependent_rollout}

The
overall uncertainty set is the Cartesian product $\mathcal{M} \coloneqq \mathcal{I} \times \mathcal{J}$,
where each model $m=(i,j)\in\mathcal{M}$ defines one complete SV
driver description by combining one gateway preset and one IDM
kinematic preset.
Prediction is then organized by
$\omega^{(i,j)}$, $t_{\mathrm{cf}}\ge 0$,
and a braking detection condition evaluated once the observer has
converged. Specifically, let
\begin{equation}
    \mathcal{B}^{(i,j)} \coloneqq
    \Bigl\{
        a_{\mathrm{idm}}(\theta^{(j)})
        + \hat\Delta_\theta
        \;<\;
        a_{\mathrm{free}}(\theta^{(j)})
        - \varepsilon
    \Bigr\},
    \label{eq:braking_condition}
\end{equation}
denote the event that, under the model pair $(i,j)$, the
observer-corrected SV acceleration falls below the corresponding
free-flow prediction by a margin $\varepsilon>0$, indicating that the
SV has begun car-following deceleration. Based on $\omega^{(i,j)}$,
$t_{\mathrm{cf}}$, and $\mathcal{B}^{(i,j)}$, three cases are
considered.

\paragraph{Case 1: $\omega^{(i,j)}=0$ for all $(i,j)\in\mathcal{M}$.}
No interaction gate is active under any model pair, so the SV remains
in free-flow for all \((i,j)\in\mathcal{M}\). The observer is inactive,
and no online correction is available.

\paragraph{Case 2: $\omega^{(i,j)}=1$ for at least one $(i,j)\in\mathcal{M}$ and
$t_{\mathrm{cf}}\le t_{\mathrm{conv}}$.}
At least one gate is active and the observer is running, but
\(\hat\Delta_\theta\) is still transient and not yet reliable.

\paragraph{Case 3: $\omega^{(i,j)}=1$ for at least one $(i,j)\in\mathcal{M}$,
$t_{\mathrm{cf}}>t_{\mathrm{conv}}$, and $\mathcal{B}^{(i,j)}$ holds.}
The observer has converged and the corrected prediction confirms active
SV deceleration for the corresponding model pair \((i,j)\). Only in
this case does prediction switch from the multi-model evaluation over
\(\mathcal{M}\) to a single observer-corrected rollout for the active
model pair, with correction \(\hat\Delta_\theta\) injected.

If \(\omega^{(i,j)}=1\) for at least one \((i,j)\in\mathcal{M}\) and
\(t_{\mathrm{cf}}>t_{\mathrm{conv}}\) but \(\mathcal{B}^{(i,j)}\) does
not hold, then the observer has converged but the SV still remains in
free-flow under the active model pair. The situation therefore reduces
to Case~1. The observer continues running, and the system switches to
Case~3 once \(\mathcal{B}^{(i,j)}\) is satisfied.

\subsubsection{Multi-Model Rollout for Cases~1 and~2}
\label{subsubsec:multistyle_rollout}

In Cases~1 and~2, all $|\mathcal{M}|$ models are rolled out from the
current state $z_0$. For each $m=(i,j)\in\mathcal{M}$, $N$ forward
steps are simulated using the joint dynamics $f$, with the SV
acceleration at step $k$ given by
$a_{\mathrm{sv},k}^{(m)}
    =
    a_{\mathrm{pidm}}\!\bigl(
        \theta^{(j)}
    \bigr).$
No observer correction is injected.
The worst-case evaluation point is selected over all trajectories
and prediction steps:
\begin{equation}
    (m^*,\,k^*)
    \;\coloneqq\;
    \arg\min_{\substack{m\in\mathcal{M},\\ k\in\{1,\dots,N\}}}
    H_{\mathrm{SV}}\!\bigl(z_k^{(m)}\bigr).
    \label{eq:kstar_robust}
\end{equation}
$J_{k^*} = \mathrm{d}z_{k^*}/\mathrm{d}z_0$ is then
propagated through the winning trajectory $\{z_k^{(m^*)}\}$ using the
parameter set of $m^*$.

\subsubsection{Observer-Corrected Rollout for Case~3}
\label{subsubsec:observer_rollout}

In Case~3, the observer has converged and SV braking is confirmed for a model pair \((i,j)\) via \(\mathcal{B}^{(i,j)}\) in
\eqref{eq:braking_condition}. A single rollout is then performed for
that pair, with observer correction injected:
\begin{equation}
    a_{\mathrm{sv},k}^{(i,j)}
    =
    a_{\mathrm{pidm}}\!\bigl(
        \theta^{(j)},\,\varphi^{(i)}
    \bigr)
    +
    \hat\Delta_\theta\,\omega_k^{(i,j)}.
    \label{eq:case3_sv_accel}
\end{equation}
The critical step \(k^*\) is selected via
\eqref{eq:kstar}, and the Jacobian in \eqref{eq:LFG_Hpred} is propagated
along the resulting corrected trajectory.



\subsubsection{Robust Predictive and Current-Position IECBF Constraints}
\label{subsubsec:robust_predictive_barrier}

This subsection applies only to Case~3. In this case, both the predictive barrier and its
current-position special case are evaluated from the
observer-corrected rollout in \eqref{eq:case3_sv_accel}.

\paragraph{Predictive IECBF}
Under the perturbed joint dynamics~\eqref{eq:perturbed_joint_dyn}, the
time derivative of \(H^{\mathrm{pred}}_{\mathrm{SV}}\) along the
true trajectory satisfies
$\dot H^{\mathrm{pred}}_{\mathrm{SV}}
    =
    L_F H^{\mathrm{pred}}_{\mathrm{SV}}
    + L_G H^{\mathrm{pred}}_{\mathrm{SV}}\,u
    + \sigma\,\Delta_\theta \ge
    \widetilde L_F H^{\mathrm{pred}}_{\mathrm{SV}}
    + L_G H^{\mathrm{pred}}_{\mathrm{SV}}\,u,
$
where $\tau\coloneqq
|\sigma|\,\bar e(t_{\mathrm{cf}})$ and $\widetilde L_F H^{\mathrm{pred}}_{\mathrm{SV}}
    \coloneqq
    L_F H^{\mathrm{pred}}_{\mathrm{SV}}
    + \sigma\,\hat\Delta_\theta$.
Accordingly, the robust predictive IECBF constraint for Case~3 is
\begin{equation}
\begin{split}
    \widetilde L_F H^{\mathrm{pred}}_{\mathrm{SV}}
    + L_G H^{\mathrm{pred}}_{\mathrm{SV}}\,u
    - \tau
    \ge
    -\alpha_{\mathrm{SV}}\!\bigl(H^{\mathrm{pred}}_{\mathrm{SV}}\bigr).
\end{split}
\label{eq:robust_cbf_constraint}
\end{equation}

\paragraph{Current-position guard}
The current-position guard is the \(N=0\) special case of the predictive
barrier, i.e., \(H_{\mathrm{SV}}^{\mathrm{curr}}=H_{\mathrm{SV}}(z_0)\). Accordingly, define $\sigma_0
    \coloneqq
    \frac{\partial H_{\mathrm{SV}}^{\mathrm{curr}}}{\partial v_s}(z_0)$, $\tau_0
    \coloneqq
    |\sigma_0|\,\bar e(t_{\mathrm{cf}})$ and $\widetilde L_F H_{\mathrm{SV}}^{\mathrm{curr}}
    \coloneqq
    L_F H_{\mathrm{SV}}^{\mathrm{curr}} + \sigma_0\,\hat\Delta_\theta$, then the robust current-position guard for Case~3 is
\begin{equation}
    \widetilde L_F H_{\mathrm{SV}}^{\mathrm{curr}}
    + L_G H_{\mathrm{SV}}^{\mathrm{curr}}\,u
    - \tau_0
    \ge
    -\alpha_{\mathrm{SV}}\!\bigl(H_{\mathrm{SV}}^{\mathrm{curr}}\bigr).
    \label{eq:robust_guard_constraint}
\end{equation}

\begin{theorem}[Robust current-position safety in Case~3]
\label{thm:robust_current_safety}
Consider Case~3. Suppose the true SV acceleration satisfies
\eqref{eq:sv_accel_uncertain} and \eqref{eq:uncertainty_bounds}, and
let \(\hat\Delta_\theta\) and \(\bar e(t_{\mathrm{cf}})\) be generated
as in Section~\ref{subsec:uncertainty_estimator}. If \(u\) is locally
Lipschitz and satisfies \eqref{eq:robust_guard_constraint} for all
\(t\ge0\), then the set
$\mathcal{C}^{\mathrm{curr}}_{\mathrm{SV}}
=
\{ z \in \mathcal{Z} \mid H_{\mathrm{SV}}^{\mathrm{curr}}(z) \ge 0 \}$
is forward invariant. Hence, if
$H_{\mathrm{SV}}^{\mathrm{curr}}(z(0)) \ge 0$ at \(t=0\), then
$H_{\mathrm{SV}}^{\mathrm{curr}}(z(t)) \ge 0$ for all \(t\ge0\).
\end{theorem}

\begin{proof}
By \eqref{eq:robust_guard_constraint},
$\dot H_{\mathrm{SV}}^{\mathrm{curr}}(z)
\ge
-\alpha_{\mathrm{SV}}\!\bigl(H_{\mathrm{SV}}^{\mathrm{curr}}(z)\bigr)$.
Since \(\alpha_{\mathrm{SV}}\) is class-\(\mathcal{K}\), the standard
CBF comparison argument yields forward invariance of
\(\mathcal{C}^{\mathrm{curr}}_{\mathrm{SV}}\)~\cite{dacs2025robust}.
\end{proof}
The predictive IECBF constraint is used as an anticipatory constraint
that accounts for the most critical predicted EV--SV configuration over
the finite horizon. Unlike the current-position guard, it is not used
here to claim a separate forward-invariance result, since
\(H_{\mathrm{SV}}^{\mathrm{pred}}\) is defined through a rollout-based
minimum and is generally nonsmooth due to switching of the critical
step \(k^*\).

\subsection{Robust CLF--IECBF QP}
\label{subsec:robust_qp}
For \(\ell\in\{\mathrm{pred},\mathrm{curr}\}\), define
\[
\widehat L_F H_{\mathrm{SV}}^\ell \coloneqq
\begin{cases}
L_F H_{\mathrm{SV}}^\ell, & \text{Cases~1--2},\\
\widetilde L_F H_{\mathrm{SV}}^\ell-\tau_\ell, & \text{Case~3},
\end{cases}
\]
with \(\tau_{\mathrm{pred}}=\tau\) and \(\tau_{\mathrm{curr}}=\tau_0\).
The robust
CLF--IECBF QP is finally given as
\begin{equation}
\label{eq:robust_clf_cbf_qp}
\begin{aligned}
u^\star
&=
\arg\min_{\substack{u\in\mathcal U,\\ \delta_y,\,\delta_\psi}}
\;
\frac{1}{2}u^\top Q u
+\frac{1}{2}p_y\delta_y^2
+\frac{1}{2}p_\psi\delta_\psi^2 \\
\text{s.t.}\quad
&\eqref{eq:clf_y_constraint},\ \eqref{eq:clf_psi_constraint},\ \eqref{eq:cbf_RU_cond},\\
&\widehat L_F H_{\mathrm{SV}}^\ell
+ L_G H_{\mathrm{SV}}^\ell\,u
\ge
-\alpha_{\mathrm{SV}}\!\bigl(H_{\mathrm{SV}}^\ell\bigr), \ell\in\{\mathrm{pred},\mathrm{curr}\}.
\end{aligned}
\end{equation}


\section{Simulation Results}
\label{sec:simulation_results}

The proposed framework is built on top of \texttt{highway-env}~\cite{highway-env}. All simulations were run on a MacBook Pro equipped with an Apple M3 chip.
%
%
The simulation is set up as described in Section~\ref{sec:sim_setup}. The representative P-IDM settings are shown in Table~\ref{tab:idm_presets} and the remaining design parameters are shown in Table~\ref{tab:all_params}. During the simulation, the true SV behavior is generated by a \emph{conservative} IDM model with a \emph{cautious} gateway, whereas the EV nominal joint model assumes an \emph{aggressive} SV IDM model with a \emph{cooperative} gateway. The EV is initialized in the upper lane at
$(X_e(0),Y_e(0))=(20,\;Y_{\mathrm{upper}})$,
and the SV is initialized in the lower lane at
$(X_s(0),Y_s(0))=(X_e(0)-5.5,\;Y_{\mathrm{lower}})$.
The RU is placed in the upper lane at
$(X_{\mathrm{RU}},Y_{\mathrm{RU}})=(X_e(0)+6,\;Y_{\mathrm{upper}})$.
The simulation sampling time is $\Delta t=0.1~\mathrm{s}$.
Here, $Y_{\mathrm{upper}}$ and $Y_{\mathrm{lower}}$ denote the lateral positions of the centers of the upper and lower lanes, respectively.
Linear class-$\mathcal{K}$ functions are used for both the CLF and CBF constraints:
$\alpha_\mathrm{RU}(s)=\kappa_\mathrm{RU} s, \alpha_\mathrm{SV}(s)=\kappa_\mathrm{SV} s, \alpha_y(s)=\kappa_y s, \alpha_\psi(s)=\kappa_\psi s,$
where $\kappa_\mathrm{SV}>0$, $\kappa_\mathrm{RU}>0$, $\kappa_y>0$, and $\kappa_\psi>0$ are design parameters. 
Baseline, nominal and robust proposed controllers are compared in the simulation. For a fair comparison, the baseline controller also adopts the predictive CBF structure, including both the current and the predictive SV safety constraint together with the RU CBF constraint. The difference is that its predictive rollout is based on the baseline non-interactive model \eqref{eq:joint_sys_baseline}.

\begin{table}[t]
\centering
\refstepcounter{table}
\label{tab:idm_presets}
\textbf{Table \thetable: P-IDM presets}\\[2pt]
\setlength{\tabcolsep}{2pt}
\renewcommand{\arraystretch}{1}
\footnotesize
\begin{tabular}{lcc lcc}
\toprule
IDM preset & $a_{\max}$ [m/s$^2$] & $b$ [m/s$^2$] & Gateway preset & $N_{\mathrm{p}}$ & $c$ [m] \\
\midrule
Conservative & 2.0 & 3.0 & Cautious    & 10 & 1.0 \\
Normal       & 4.0 & 5.0 & Normal      & 20 & 2.0 \\
Aggressive   & 6.0 & 6.0 & Cooperative & 40 & 3.0 \\
\bottomrule
\end{tabular}
\vspace{0.3em}

\footnotesize All presets share $s_0=10$\,m, $T=1.5$\,s, $\delta=4$, and $v_0=10$\,m/s.
\end{table}

\begin{table}[t]
\centering
\refstepcounter{table}
\label{tab:all_params}
\textbf{Table \thetable: Simulation, controller, and robust-design parameters.}\\[2pt]
\setlength{\tabcolsep}{2pt}
\renewcommand{\arraystretch}{1.0}
\footnotesize
\begin{tabular}{cccccc}
\toprule
Parameter & Value & Parameter & Value & Parameter & Value \\
\midrule
$v_e(0)$                 & $10$\,m/s       & $p_y$                   & $25$           & $\lambda_s$       & $10$\,s$^{-1}$ \\
$v_s(0)$                 & $12.5$\,m/s     & $p_\psi$                & $15$           & $\delta_b$        & $0.5$\,m/s$^2$ \\
$r^\mathrm{RU}_{a}, r^\mathrm{RU}_{b}$      & $2.0$\,m        & $\kappa_y,\kappa_\psi$           & $1.5$          & $\delta_L$        & $0.5$\,m/s$^3$ \\
$r^\mathrm{SV}_{a}$      & $4.5$\,m        & $\kappa_{\mathrm{SV}}$  & $5$            & $\delta_{e,\min}$ & $-1.8$\,rad \\
$r^\mathrm{SV}_{b}$      & $2.5$\,m        & $N$                     & $20$           & $\delta_{e,\max}$ & $+1.8$\,rad \\
\bottomrule
\end{tabular}
\end{table}

Figures~\ref{fig:control_inputs} compare the baseline, nominal proposed, and robust proposed controllers. The baseline controller becomes infeasible immediately. As a result, both the steering and acceleration inputs remain zero, the EV cannot initiate an effective avoidance maneuver, and the vehicle continues toward the RU. This is further confirmed by the RU barrier value $H_{\mathrm{RU}}$ which drops below zero at around $0.4\,\mathrm{s}$, indicating loss of safety. The corresponding global trajectory in Figure~\ref{fig:global_trajectories} shows that the EV eventually collides with the RU obstacle.
In contrast, both the nominal proposed and robust proposed controllers remain feasible throughout the maneuver and generate feasible control actions from solving the QP, as shown in Figure~\ref{fig:control_inputs}.
The difference between the nominal proposed and robust proposed controllers is more clearly seen in the global trajectories in Figure~\ref{fig:global_trajectories}. For the nominal proposed controller, the internal joint model assumes a more cooperative SV response than the true one. Because of this model mismatch, although the EV completes the maneuver safely and the barrier values remain nonnegative, the true SV exhibits a small lateral deviation, visible in its trajectory. This deviation is generated by the built-in background-vehicle behavior in \texttt{highway-env}, where SVs follow IDM longitudinal dynamics together with a MOBIL-based lateral policy~\cite{highway-env}. By contrast, the robust proposed controller accounts for this mismatch through the multi-model predictive rollout and observer-based tightening, and therefore avoids this undesirable lateral displacement of the SV while still maintaining feasibility and safety.
Figure~\ref{fig:control_inputs} shows that the nominal proposed and robust proposed controllers keep the barrier values nonnegative throughout the maneuver, whereas the baseline controller violates the RU safety constraint. These results show that the proposed interaction-aware predictive CBF improves feasibility relative to the baseline model, and that the robust extension further improves the physical consistency of the interaction.

\begin{figure}[t]
  \centering
  \begin{minipage}[t]{0.48\columnwidth}
    \centering
    \includegraphics[width=\columnwidth]{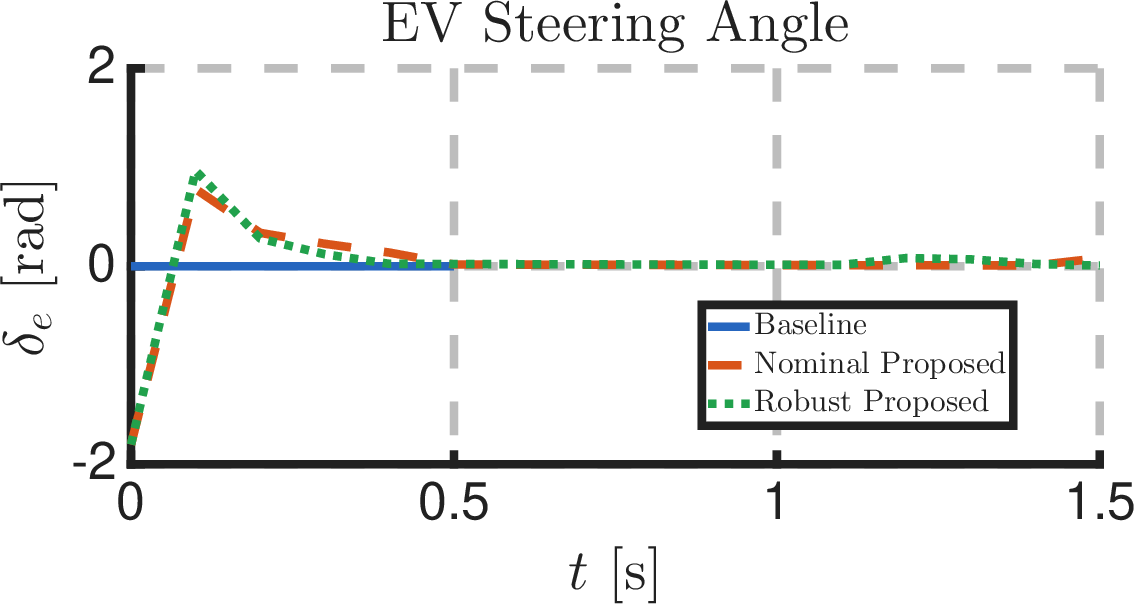}
  \end{minipage}
  \hfill
  \begin{minipage}[t]{0.48\columnwidth}
    \centering
    \includegraphics[width=\columnwidth]{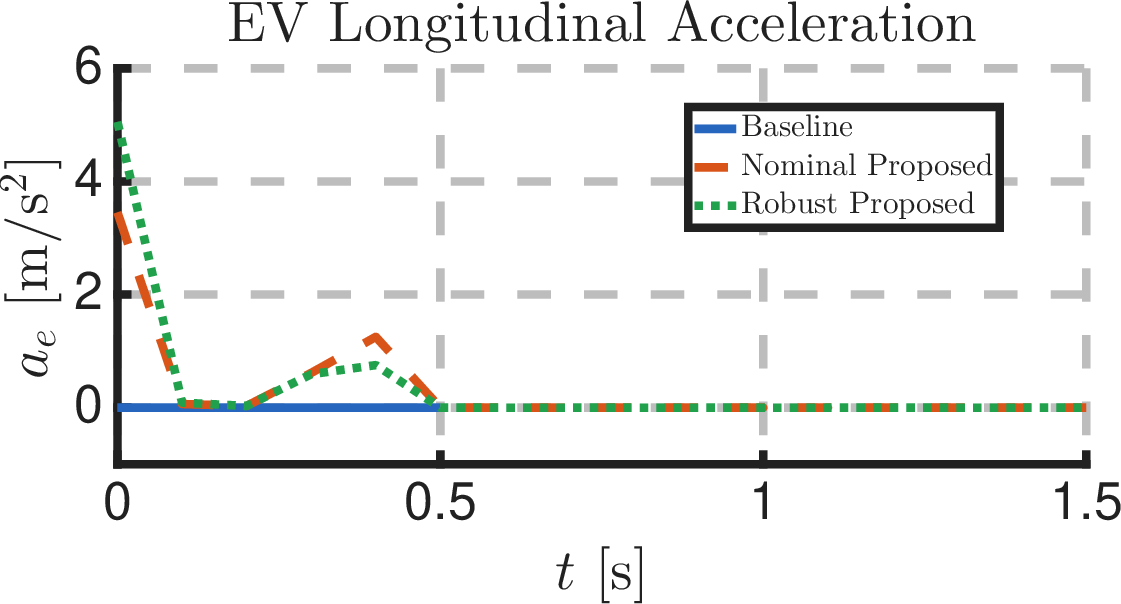}
  \end{minipage}
  \vspace{2pt}
  \includegraphics[width=0.48\columnwidth]{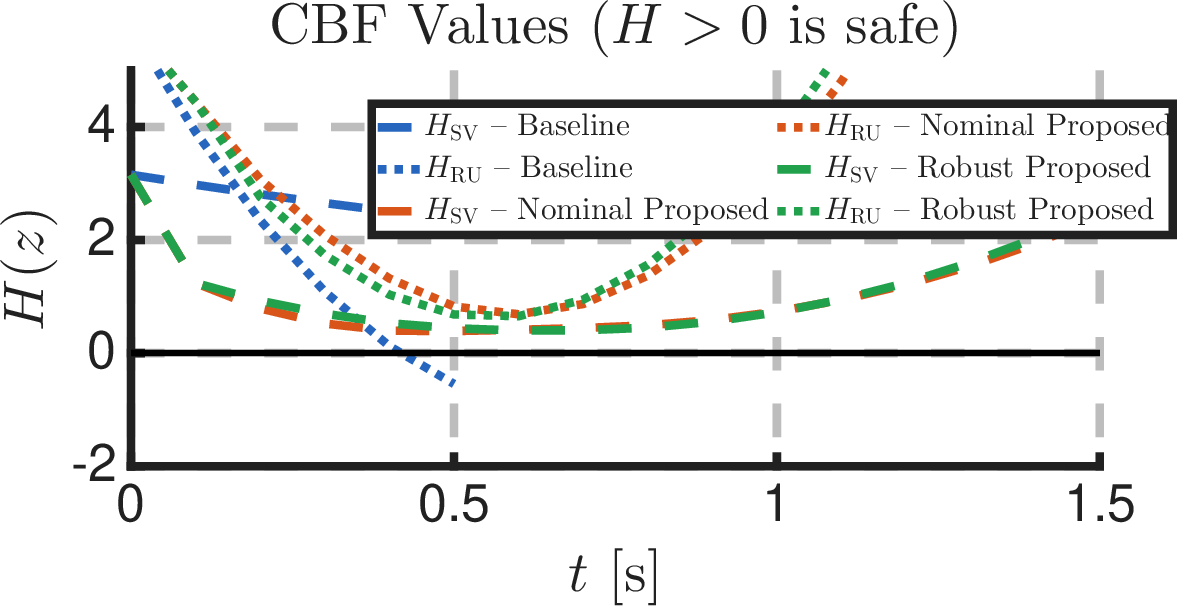}
  \caption{EV control inputs $\delta_e$ and $a_e$ and CBF values generated by the baseline, nominal proposed, and robust proposed controllers.}
  \label{fig:control_inputs}
\end{figure}

\begin{figure}[t]
  \centering
  \includegraphics[width=0.48\columnwidth]{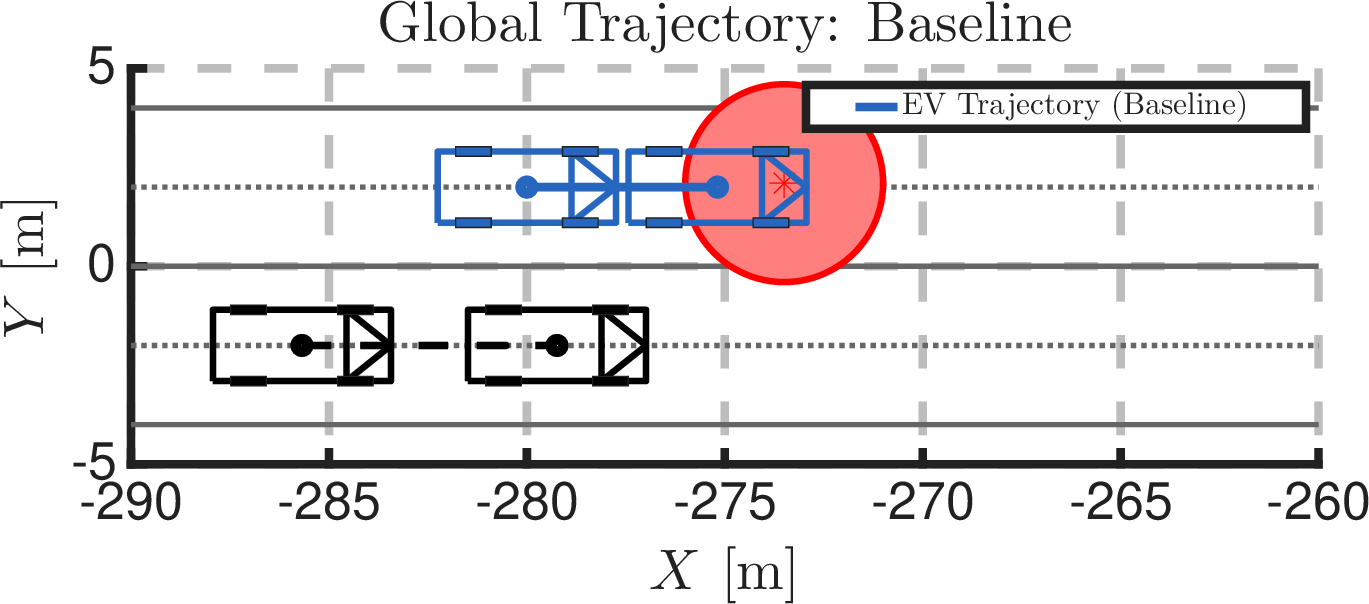}

  \begin{minipage}[t]{0.48\columnwidth}
    \centering
    \includegraphics[width=\linewidth]{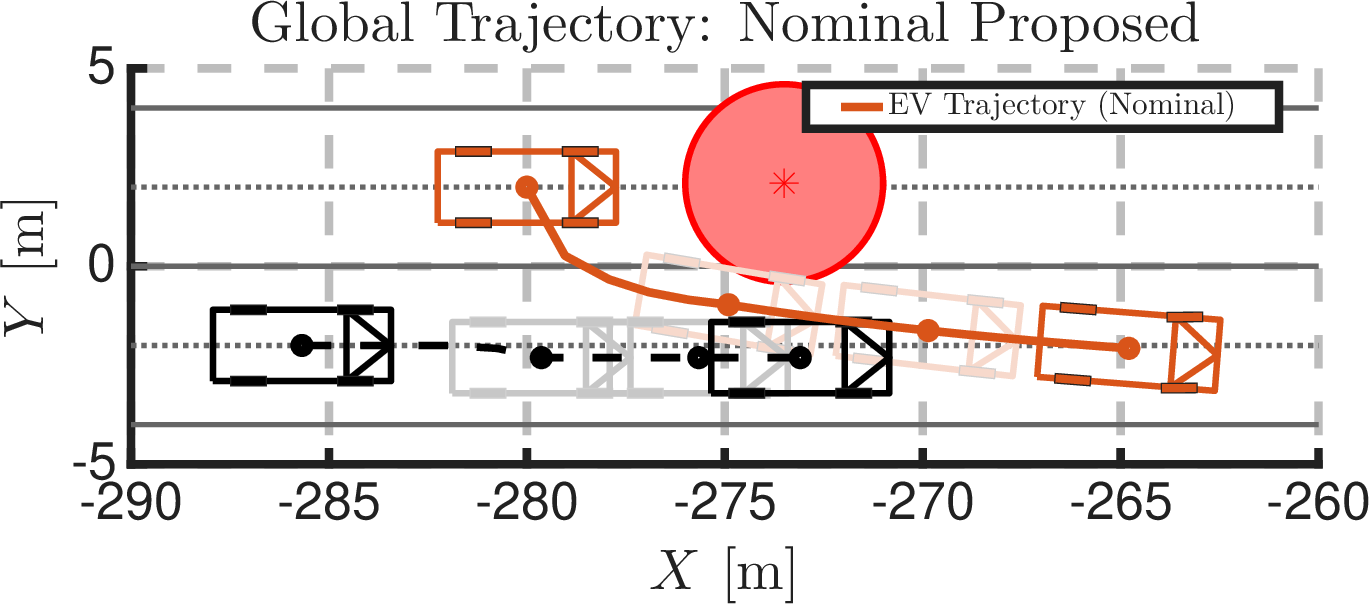}
  \end{minipage}
  \hfill
  \begin{minipage}[t]{0.48\columnwidth}
    \centering
    \includegraphics[width=\linewidth]{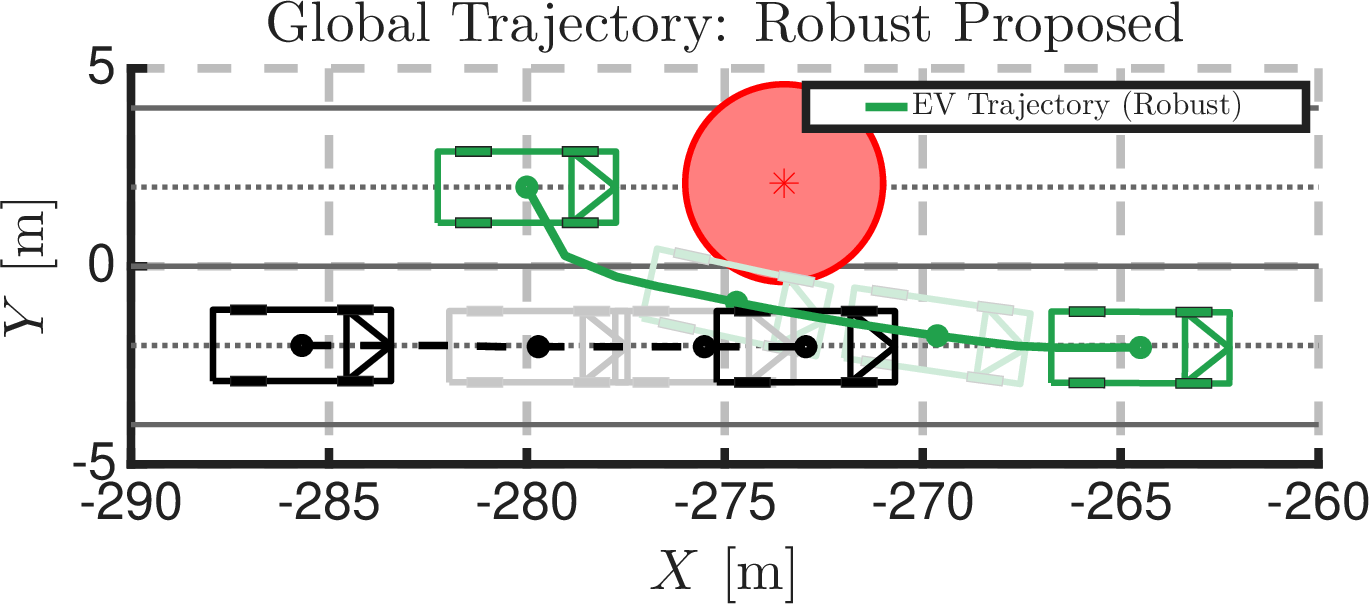}
  \end{minipage}
  \caption{Global trajectories for the baseline, nominal proposed, and robust proposed controllers. In all three plots, the SV is the black vehicle in the lower lane and the RU is the red circular obstacle.}
  \label{fig:global_trajectories}
\end{figure}

\section{Conclusions}
This paper proposed an IECBF for emergency lane change in mixed traffic,
incorporating an interaction-aware joint EV--SV model with P-IDM into
the ECBF framework. As demonstrated in simulation, compared with classical ECBF neglecting SV reactions,
the proposed method improves feasibility and safety. A robust extension combining multi-model rollout with
observer-based correction addresses SV-model uncertainty, further
improving feasibility under model mismatch.

\bibliographystyle{IEEEtran}

\bibliography{references}

\end{document}